\documentclass[12pt]{article}

\usepackage[utf8]{inputenc}

\usepackage{amsmath,amsthm,amssymb}
\usepackage{mathrsfs}
\usepackage{authblk}

\theoremstyle{plain}
\newtheorem{theorem}{Theorem}

\newtheorem{proposition}{Proposition}
\theoremstyle{definition}

\theoremstyle{remark}
\newtheorem{remark}{Remark}
\newtheorem*{remark*}{Remark}
\newtheorem{example}{Example}

\numberwithin{equation}{section}

\usepackage[margin=15mm]{geometry}

\usepackage{fancyhdr}
\begin{document}

\title{Topology of the cone of positive maps on qubit systems}

\author[1]{Marek Miller\thanks{marek.miller@ift.uni.wroc.pl}}
\author[1]{Robert Olkiewicz\thanks{robert.olkiewicz@ift.uni.wroc.pl}}
\affil[1]{Instytut Fizyki Teoretycznej, Uniwersytet Wrocławski, Poland}

\date{}

\maketitle

\abstract{
An alternative, geometrical proof of a known theorem concerning
the decomposition of positive maps of the matrix algebra $M_{2}(\mathbb{C})$
has been presented.
The premise of the proof is the identification of positive maps
with operators preserving the Lorentz cone in four dimensions,
and it allows to decompose the positive maps with respect to those 
preserving the boundary of the cone.
In addition, useful conditions implying complete positivity
of a map of $M_{2}(\mathbb{C})$ have been given,
together with a sufficient condition for complete positivity of
an extremal Schwarz map of $M_{n}(\mathbb{C})$.
Lastly, following the same geometrical approach, a description
in topological terms of maps that are simultaneously completely positive
and completely copositive has been presented.
}

\section*{Introduction}
\label{sec:Introduction}

\paragraph{}
The theory of positive maps constitutes an increasingly popular subject of research,
and its mathematical appeal matches only its ability to present constructive
applications to the theory of quantum information and quantum entanglement
\cite{stormer2012positive, chruscinski2014entanglement}.
The classification of positive maps on $M_{2}(\mathbb{C})$,
the algebra of square complex matrices of size 2,
in physical terms representing a qubit,
was a first turning point of the theory.
The theorem proven by E.\,Størmer 
\cite{stormer1963positive}
and S.\,L.\,Woronowicz
\cite{woronowicz1976positive}
says that every positive map
$S: M_{2}(\mathbb{C}) \rightarrow M_{n}(\mathbb{C})$,
where $n \leq 3$, can be decomposed into the sum:
$S = \Lambda_{1} + \Lambda_{2} \circ t$,
where both maps $\Lambda_{1}, \Lambda_{2}: M_{2}(\mathbb{C}) \rightarrow M_{n}(\mathbb{C})$
are completely positive, possibly zero,
and $t$ stands for the transposition map of $M_{2}(\mathbb{C})$.
As the structure of completely positive maps had been entirely understood
thanks to the work of 
W.\,F.\,Stinespring
\cite{stinespring1955positive}
and M.-D.\,Choi
\cite{choi1975completely},
the theorem offered the full description of positive maps in this low-dimensional case.
Hence, e.g. the structure of entanglement witnesses of two-qubit systems can be
expressed concisely with the celebrated PPT criterion by
A.\,Peres 
\cite{peres1996separability}
and R., P., M. Horodecki \cite{horodecki1996separability}.

In this paper, we present an alternative proof of the cited theorem.
The premise of the proof is the identification of positive maps
of $M_{2}(\mathbb{C})$ with operators preserving the four-dimensional Lorentz cone.
We show that the decomposition of positive maps is rooted in the fact that
every such operator is a convex combination of those additionally preserving
the cone's topological boundary 
\cite{loewy1975positive}.
Thus, the proof of the theorem,
contrary to the previous attempts based mainly
on involved algebraic manipulation of matrices,
provides a geometrical insight into the nature of the decomposition discussed above.
It must be mentioned that a similar method of looking at positive maps and their
application to the theory of quantum composite systems has been employed before
(see e.g. \cite{leinaas2006geometrical}).
In addition, we approach the conjecture stated by Robertson
\cite{robertson1983schwarz},
who asks whether a positive extremal map that fulfils the Kadison-Schwarz
inequality \eqref{eq:SchwarzInequality} is completely positive.
With an additional assumption, we give the affirmative answer to this question
for positive maps on matrix algebras.
Last but not least, using that geometrical approach,
we have been able to provide useful conditions for a positive map to be completely positive,
as well as to describe the set of bistochastic maps that are simultaneously
completely positive and completely copositive in terms of
a norm-induced topology on the unit ball of the identified set of operators.

\section{Preliminaries}
\label{sec:Preliminaries}

\paragraph{}
Let
$M_{n} = M_{n}(\mathbb{C})$ be the algebra of complex matrices
of size $n$;
the algebra of real matrices will always be denoted explicitly by
$M_{n}(\mathbb{R})$.
We will use letters $A, B, X$, etc. to specify a matrix of $M_{n}$.
The symbol $\mathbf{1}_{n}$, or simply $\mathbf{1}$,
means the identity matrix of $M_{n}$.
The norm of $A$, denoted by $||A||$, is understood as the operator
norm of $A$ as a linear map acting on $\mathbb{C}^{n}$.
For $A \in M_{n}$,
we denote its trace by $\text{Tr} \, A$;
and by $A^{t}$ and $A^{*} = \overline{A}^{\,t}$
its transpose and conjugate transpose, respectively.
The commutator of two matrices $A, B$ is defined as
$[A, \, B] = AB - BA$.
We say that a matrix $A$ is positive-semidefinite,
or simply positive,
if $A = A^{*}$ and $A$ has a non-negative spectrum.

A linear map $S\!: M_{n} \rightarrow M_{n}$ is said to be positive,
indicated as $S \geq 0$,
if for any $A \in M_{n}$ such that $A \geq 0$,
we have $S(A) \geq 0$.
The identity map of $M_{n}$ is labelled $I_{n}$,
or simply $I$.
For a positive map $S$,
its operator norm is given by $||S|| = S(\mathbf{1})$.
Any positive map is Hermitian,
i.e.  $S(X^{*}) = S(X)^{*}$, for all $X \in M_{n}$.
We say that a positive map fulfils the Kadison-Schwarz inequality, if
\begin{equation}
\label{eq:SchwarzInequality}
    S(X^{*}) \, S(X) \: \leq \: S(X^{*} X),
\end{equation}
for any matrix $X \in M_{n}$.
A positive map $S$ that preserves the identity,
i.e. $S(\mathbf{1}) = \mathbf{1}$,
fulfils the Kadison-Schwarz inequality, but only for normal matrices
$X$, that is the ones such that $X X^{*} = X^{*} X$
\cite{choi1974schwarz}.
Folowing \cite{robertson1983schwarz}, we call a positive map that
preserves the identity and fullfils the Kadison-Schwarz inequality --
a Schwarz map.
A map $S$ is called bistochastic,
if it is positive
and preserves both the identity and trace,
i.e.
$S(\mathbf{1}) = \mathbf{1}$
and
$\text{Tr} \, S(X) = \text{Tr} \, X$
for any $X \in M_{n}$.
For $k \in \mathbb{N}$, a positive map $S$,
such that the map
$I_{k} \otimes S :  M_{k} \! \otimes \! M_{n}
 \rightarrow  M_{k} \! \otimes \! M_{n}$
is positive,
is called $k$-positive.
If a map is a $k$-positive map for every $k$,
it is called completely positive.
Similarly, a map is $k$-\emph{co}positive,
or completely copositive,
if $I_{k} \otimes (S \circ t)$ is positive for
some $k$, or for every $k$, respectively,
where $t \! : A \mapsto A^{t}$,
$A \in M_{n}$, is the transposition map.
A positive map $S$ is called decomposable,
if it can be written in the form
$S(A) = \Lambda_{1}(A) + \Lambda_{2} (A^{t})$,
where both maps $\Lambda_{1}, \Lambda_{2}$ are completely positive,
possibly zero
(see \cite{choi1975completely} for more details on completely positive maps).
By an extremal positive map on $M_{n}$,
we exclusively mean an extremal element of the cone of all positive maps on $M_{n}$,
as defined below.

From now on, we focus on the algebra $M_{2}$.
We shall use the Greek indices $\mu, \nu$ to denote the numbers
$\mu,\nu = 0,1,2,3$;
whereas the Latin indices stand for
$i,j = 1,2,3$.
Let $(\sigma_{\mu})$, $\mu = 0, 1, 2, 3$;
be a basis of $M_{2}$
such that $\sigma_{0} = \tfrac{1}{\sqrt{2}}  \mathbf{1}$ and
$\sigma_{i}$, $i=1,2,3$,
are normalised Pauli matrices
 \begin{equation}
  \sigma_{1} = \frac{1}{\sqrt{2}} \begin{pmatrix}
            0 & 1 \\ 1 & 0
               \end{pmatrix},
 \quad
 \sigma_{2} = \frac{1}{\sqrt{2}}  \begin{pmatrix}
            0 & -i \\ i & 0
               \end{pmatrix},
 \quad
 \sigma_{3} = \frac{1}{\sqrt{2}}  \begin{pmatrix}
            1 & 0 \\ 0 & -1
               \end{pmatrix}.
 \end{equation}
A Hermitian element $\sigma(x)$ of $M_{2}$,
$\sigma(x) = \sum_{\mu =0}^{3} x_{\mu} \sigma_{\mu}$,
for $x = (x_{\mu}) = (x_{0}, \vec{x}) \in \mathbb{R}^{4}$,
is positive if and only if $x \in L_{4}$,
where
$ L_{4} = \left \{
            x \in \mathbb{R}^{4} :
            x_{0} \geq \sqrt{(x_{1})^{2} + (x_{2})^{2} + (x_{3})^{2}} \,
        \right \}$.
By $\vec{x} = (x_{1}, x_{2}, x_{3}) \in \mathbb{R}^{3}$,
we denote the 'spatial' part of the vector $x$,
and by $||\vec{x}||$ -- its Euclidean norm.

Let $S: M_{2} \rightarrow M_{2}$ be a linear map.
If $S$ is Hermitian,
then we can define a real matrix $\pi(S) \in M_{4}(\mathbb{R})$ by
\begin{equation}
\label{RandomLabel:365828}
  \sigma(\pi(S) x) = S \sigma(x), \quad x \in \mathbb{R}^{4}.
\end{equation}
Clearly, for two Hermitian maps $S_{1}, S_{2}$, we have
$\pi(S_{1} S_{2}) = \pi(S_{1}) \pi(S_{2})$
and $\pi(S^{-1}) = \pi(S)^{-1}$, whenever the map $S^{-1}$ exists.
The map $S$ is positive,
if and only if
$\pi(S)$ maps $L_{4}$ into itself.

A set $\mathcal{P} \subset \mathbb{R}^{n}$ is a
(convex) cone,
if $\alpha x + \beta y \in \mathcal{P}$,
whenever $x, y \in \mathcal{P}$,
and $\alpha, \beta \geq 0$.
An element $x \in \mathcal{P}$ is called extremal,
if from the fact that $x - y \in \mathcal{P}$,
for some vector $y \in \mathcal{P}$,
follows that $y = \alpha x$, $\alpha \geq 0$.
We will denote by Ext\,$\mathcal{P}$ the set of extremal elements of
the cone $\mathcal{P}$.
The set of operators  of $M_{n}(\mathbb{R})$ that map $\mathcal{P}$ into itself
will be denoted by $\Gamma(\mathcal{P})$.
For any subset $U$ of a linear space, by the convex hull of $U$
we mean the set
$\text{conv} U = \left \{ \, tx + (1-t)y;
\,\, x,y \in U, \, 0 \leq t \leq 1 \right \}$
containing all convex combinations of the points in $U$.
All the subsets of the set of positive maps specified above,
i.e. the set of bistochastic, k-positive, k-copositive, completely positive,
completely copositive and decomposable maps,
are convex cones themselves.

We call the cone $L_{4}$ the Lorentz cone
in $\mathbb{R}^{4}$.
The topological boundary $\partial L_{4}$ of $L_{4}$ is the set
$\partial L_{4} = \{ x \in L_{4} : x_{0} = || \vec{x} || \}$.
It it true that in this case $\partial L_{4} = \text{Ext}\,L_{4}$.
The symbol $\Theta(L_{4}) \subset \Gamma(L_{4})$
will stand for the group
of operators of $M_{4}(\mathbb{R})$
that map $\partial L_{4}$ onto itself.
By \cite[Theorem 2.4]{loewy1975positive},
for every $\tilde{A} \in \Theta(L_{4})$,
$\tilde{A} = r \tilde{O}$,
where $r > 0$ and $\tilde{O}$ belongs to
the orthochronous Lorentz group
$\mathrm{O}^{+}(1,3)$.
Let $\rho: SL_{2}(\mathbb{C}) \rightarrow \text{SO}^{+}(1,3)$
be the standard homomorphism (the spinor map)
between the group of invertible
complex matrices of $M_{2}$
with determinant equal to one
and the \emph{proper} orthochronous Lorentz group $\text{SO}^{+}(1,3)
\subset \text{O}^{+}(1,3)$
(see e.\,g. \cite{naber1992geometry} for a definition of the Lorentz group).
By definition, 
\begin{equation}
\sigma(\rho(V) x) = V^{*} \sigma(x) V, 
\end{equation}
for any $V \in SL_{2}(\mathbb{C})$ and $x \in \mathbb{R}^{4}$.
Hence, for a map $S_{V}: M_{2} \rightarrow M_{2}$,
such that $S_{V} A = V^{*} A V$, $V \in \text{SL}_{2}(\mathbb{C})$,
we have
$\rho(V) = \pi(S_{V})$.
It is evident that every element of $\text{O}^{+}(1,3)$,
which does not belong to $\text{SO}^{+}(1,3)$,
can be written as a multiple $\Lambda J$,
where $\Lambda \in \text{SO}^{+}(1,3)$ and $J$  is a diagonal matrix,
$J = \mathrm{diag}(1,1,-1,1)$.
(The particular form of the matrix $J$ has been chosen
only for the sake of convenience;
the essential fact is that 
$J \in \text{O}^{+}(1,3)$ and
$\mathrm{det} \, J = -1$.)

\section{Positive maps}
\label{sec:PositiveMaps}

We present the main result of the paper,
which is a geometrical proof of the decomposition theorem for positive
maps of $M_{2}$.

\begin{theorem}
\label{thm:PositiveMaps}
    Let $S:M_{2} \rightarrow M_{2}$ be a positive map.
    Then $S = \Lambda_{1}  + \Lambda_{2} \circ t$,
    where the maps
    $\Lambda_{1}, \Lambda_{2}:M_{2} \rightarrow M_{2}$
    are completely positive
    and $t$ stands for the transposition map of $M_{2}$.
\end{theorem}

\begin{proof}
\label{RandomLabel:875919}
  The map $S$ acts on the basis vectors $\sigma_{\mu}$ by
$S \sigma_{\mu} = \sum_{\nu = 0}^{3} \pi(S)_{\mu \nu} \sigma_{\nu}$.
Suppose at first that $\pi(S) \in \Theta(L_{4})$. 
If $\pi(S) = \rho(V) \in \text{SO}^{+}(1,3)$,
$V \in SL_{2}(\mathbb{C})$,
then
$S \sigma(x) = \sum_{\mu = 0}^{3}  x_{\mu} V^{*} \sigma_{\mu} V =
 V^{*} \sigma(x) V$,
 $x \in \mathbb{R}^{4}$,
and the map $S$ is completely positive,
since every matrix $A \in M_{2}$ is a linear combination
of Hermitian matrices.
On the other hand,
if $\pi(S) = \rho(V) J$,
then
\begin{multline}
\label{RandomLabel:830200}
S \sigma(x) =
  x_{0} \, V^{*}V + x_{1} \, V^{*} \sigma_{1} V -
  x_{2} \, V^{*} \sigma_{2} V +
  x_{3} \, V^{*} \sigma_{3} V = \\
  = x_{0} \, V^{*}V + x_{1} \, V^{*} \sigma_{1}^{t} V +
  x_{2} \, V^{*} \sigma_{2}^{t} V +
  x_{3} \, V^{*} \sigma_{3}^{t} V =
  V^{*} \sigma(x)^{t} V,
\end{multline}
because $\sigma_{2}^{t} = - \sigma_{2}$.
Hence, $S$ is a composition of the transposition map and the one that is
completely positive
(i.e. $S$  is completely \emph{co}positive).
Next, suppose that $\pi(S) \in \delta(L_{4})$,
where
$
 \delta(L_{4}) = \left \{ u w^{t}:
 \, u, w \in \partial L_{4}  \right \}
$,
the set of rank-one operators that preserve $\partial L_{4}$.
It follows from 
\cite{loewy1975positive}, Lemma 3.2, that
$S$ is an extremal positive map such that $\text{rank} S = 1$,
which implies that $S$ must be completely positive.

If $\pi(S)$ is any matrix that preserves the Lorentz cone,
we infer from
\cite{loewy1975positive}, Theorem 4.5,
that $\pi(S) \in \text{conv} \left ( \Theta(L_{4}) \cup \delta(L_{4}) \right)$. 
This deep geometrical result can be described by saying that every
operator that preserves the Lorentz cone $L_{4}$ is a convex combinations
of those operators that additionally preserve the boundary of the cone.
By what has been said above,
we are allowed to write
\begin{equation}
\label{RandomLabel:587827}
    S = \Lambda_{1} + \Lambda_{2} \circ t,
\end{equation}
where each $\Lambda_{1}$, $\Lambda_{2}$
is a completely positive map.
\end{proof}

\begin{remark}
Theorem \ref{thm:PositiveMaps} provides an alternative
proof of a known result
\cite{stormer1963positive, woronowicz1976positive}.
But whereas methods employed in those papers
are based mainly on algebraic manipulations of matrices,
the introduction of the map $\pi$ has allowed us
to consider geometrical properties of the relevant
cone as a subset of  $\mathbb{R}^{4}$,
in order to obtain the full characterisation of positive maps of $M_{2}$.
This characterisation could therefore be understood
as having its source in the particular symmetry of the cone
$L_{4}$
representing positive matrices of $M_{2}$.
In this simplest case,
we were able to apply the deep result from
\cite{loewy1975positive}, concerning
the description of maps that preserve the cone $L_{4}$,
in hopes of  transforming it back to the realm of linear
maps of complex matrices.
If we switch to the case of positive maps of $M_{3}$,
the cone of those vectors
$x \in \mathbb{R}^{9}$ such that
$\lambda(x) =  \sum_{\mu=0}^{8} x_{\mu} \lambda_{\mu}$
is a positive matrix of $M_{3}$,
where $(\lambda_{\mu})_{\mu=0}^{8}$ is any linear space basis of $M_{3}$,
is no longer equal to the Lorentz cone $K_{9}$.
In fact, the geometry of the cone becomes highly non-trivial,
rendering the analysis of operators that preserve it much more 
\mbox{involved \cite{goyal2011geometry}}.
Moreover, since that cone is only a proper subset of $K_{9}$,
the \mbox{Theorem 4.5} \mbox{of \cite{loewy1975positive}} can no longer be applied.
\end{remark}

A natural question that arises from the Theorem \ref{thm:PositiveMaps}
is when a positive map of $M_{2}$ is completely positive.
Checking that by means of the definition of complete positivity might be a tedious
task, hence we propose the following result.

\begin{proposition}
\label{prop:MapsPreservingIdentity}
Let $S: M_{2} \rightarrow M_{2}$ be a positive map preserving the identity.
We consider the following conditions:
\begin{enumerate}

\item
\label{lem:condProj}
$S$ is a Schwarz map
and there exists a
rank-one orthogonal projection $P$, for which
$S(P)^{2} = S(P)$.

\item
\label{lem:condCommut}
There is a rank-one orthogonal projection $P$
such that for every $X \in M_{2}$,
$S([P,X]) = [S(P), \, S(X)]$.

\item
\label{lem:condCP}
$S$ is completely positive.
\end{enumerate}
Then the following relation holds between the above conditions:
$
\ref{lem:condProj} \Rightarrow
    \ref{lem:condCommut} \Rightarrow \ref{lem:condCP}.
$
\end{proposition}

\begin{proof}
$\ref{lem:condProj}) \Rightarrow \ref{lem:condCommut})$.
Let $A_{k} = kP + iX$, for some $X = X^{*} \in M_2$ and
$k \in \mathbb{N}$.  Putting $A_{k}$ into the Kadison-Schwarz inequality
\eqref{eq:SchwarzInequality}, we obtain
\begin{equation}
i [S(P), \,  S(X)] - i S([P, \, X])  \leq \frac{1}{k} \left( S(X^{2}) - S(X)^{2} \right),
\end{equation}
for every $k \in \mathbb{N}$.
Hence
\begin{equation}
\label{ineq:Commutators}
i [S(P), \,  S(X)] - i S([P, \, X]) \leq 0.
\end{equation}
By taking $A'_{k} =kP - i X$, and repeating essentially the same
calculation, we get the reverse of \eqref{ineq:Commutators},
and thus
\begin{equation}
\label{eq:Commutators}
S([P, X]) = [S(P), S(X)],
\end{equation}
for any Hermitian matrix $X$.
Because \eqref{eq:Commutators} is linear in $X$,
the equation holds for every matrix $X \in M_{2}$.

$\ref{lem:condCommut}) \Rightarrow \ref{lem:condCP})$.
Let $P$ be a rank-one orthogonal projection such that
\begin{equation}
\label{eq:Commutator}
S([P,X]) = [S(P), \, S(X)],
\end{equation}
for every $X \in M_{2}$.

If $S$ is a completely positive map,
then for every unitary matrices $U, V \in M_{2}$
the map $\tilde{S}: A \mapsto U^{*} S(V^{*} A V) U$
is also completely positive.
Therefore, we can assume without loss of generality that
$P = P_{1} =
\left(
\begin{smallmatrix} 1 & 0 \\ 0 & 0 \end{smallmatrix}
\right)$,
and $S(P)$ is diagonal.
Let $P_{2} = \left(
\begin{smallmatrix} 0 & 0 \\ 0 & 1 \end{smallmatrix}
\right)$,
$E_{12} =
\left(
\begin{smallmatrix} 0 & 1 \\ 0 & 0 \end{smallmatrix}
\right)$,
and
$E_{21} =
\left(
\begin{smallmatrix} 0 & 0 \\ 1 & 0 \end{smallmatrix}
\right)$.
Suppose that
$S(P_{1}) = \lambda_{1} P_{1} + \lambda_{2} P_{2}$,
$0 \leq \lambda_{2} \leq \lambda_{1} \leq 1$.
If $S(E_{12}) = 0$,
then of course $S(E_{21}) = 0$ as well, and as a result,
$S$ maps $M_{2}$
into the Abelian algebra of diagonal matrices.
Hence, $S$ is completely positive.
On the other hand, if
$S(E_{12}) =
\left(
\begin{smallmatrix} w & z \\ x & y \end{smallmatrix}
\right) \neq 0$,
$w, z, x, y \in \mathbb{C}$,
when we plug $X = E_{12}$ into \eqref{eq:Commutator},
remembering that $\lambda_{1} \geq \lambda_{2}$,
we obtain that $x = y = w = 0$
and $\lambda_{1} - \lambda_{2} = 1$.
It means that $\lambda_{1} = 1$ and $\lambda_{2} = 0$,
and $S(E_{12}) = z E_{12}$.
Altogether, $S$ acts on a matrix
$X \in M_{2}$, $X = (x_{ij})_{i,j=1,2}$, by
\begin{equation}
S(X) = \begin{pmatrix}
 x_{11} & z \, x_{12} \\
\overline{z} \, x_{21} & x_{22}
\end{pmatrix}.
\end{equation}

Now, it is easy to see that $S$ is completely positive.
Indeed,
it comes from positivity of $S$ that $|z| \leq 1$.
We write
\begin{equation}
S(X) =
\begin{pmatrix}
1 & z \\ \overline{z} & 1
\end{pmatrix} \circ X,
\end{equation}
where the symbol $\circ$ in this case denotes the Hadamard product,
i.e. the element-wise product of two matrices of the same dimension.
Because $|z| \leq 1$, the matrix 
$
\left( 
\begin{smallmatrix} 1 & z \\ \overline{z} & 1 \end{smallmatrix}
\right)$
is positive.
It is a known fact that a map
$X \mapsto A \circ X$ is completely positive,
if and only if $A$ is a positive matrix
(see e.g. Lemma 1 of \cite{besenyei2011completely} for a simple proof of this fact).
\end{proof}

\begin{example}
One can easily show that the condition \ref{lem:condProj} of
Proposition \ref{prop:MapsPreservingIdentity} is strictly stronger than the
condition \ref{lem:condCommut}.
Let $S$ be a bistochastic map of $M_{2}$ such that
\begin{equation}
\label{RandomLabel:514431}
    \pi(S) = \begin{pmatrix}
    1 & 0 & 0 & 0 \\
    0 & b & 0 & 0 \\
    0 & 0 & 0 & 0 \\
    0 & 0 & 0 & 0
    \end{pmatrix},
    \quad 0 < b < 1,
\end{equation}
i.e. $S(\sigma_{1}) = b \sigma_{1}$ and
$S(\sigma_{2}) = S(\sigma_{3}) = 0$.  
Then $S\left( [P_{0}, \, X ] \right) = [ S(P_{0}), \, S(X) ]$,
where
$P_{0}= \frac{1}{2} \left( \begin{smallmatrix}
 1 & 1 \\ 1 & 1
 \end{smallmatrix} \right)$, and $X \in M_{2}$.
However, there is no rank-one orthogonal projection $P$
such that $S(P)^{2} = S(P)$.
Indeed,
since
$P_{0} = \frac{1}{\sqrt{2}}(\sigma_{0} + \sigma_{1})$,
$S(P_{0}) = \frac{1}{\sqrt{2}}(\sigma_{0} + b \sigma_{1})$.
Let $X$ be a Hermitian matrix, $X = \sigma(x)$, $x \in \mathbb{R}^{4}$,
$S(X) = x_{0} \sigma_{0} + b x_{1} \sigma_{1}$.
Then clearly $[S(P_{0}), \, S(X) ] =0$. 
Moreover,
$[P_{0}, \, \sigma(x)] = 
i(x_{2} \sigma_{3} - x_{3} \sigma_{2})$,
hence 
$S([P_{0}, \, \sigma(x)]) = 0$.
The condition \ref{lem:condCommut} of Prop. \ref{prop:MapsPreservingIdentity}
holds for any Hermitian $X$, therefore it holds also for any matrix $X \in M_{2}$.  
On the other hand, 
$P = \sigma(a)$, $a \in \mathbb{R}^{4}$, is a rank-one orthogonal projection,
if and only if $a_{0} = ||\vec{a}|| = \frac{1}{\sqrt{2}}$.
Then $S(P) \neq \mathbf{1}$ and since $b<1$,
$S(P)$ cannot be a rank-one orthogonal projection,
i.e. $S(P)^{2} \neq S(P)$.    
\end{example}

For extremal positive maps, the condition \ref{lem:condProj} of
Proposition \ref{prop:MapsPreservingIdentity} could be proven sufficient for
complete positivity in any dimension.
At the same time, a more general problem posed by Robertson \cite{robertson1983schwarz},
whether there is a Schwarz map that is extremal as a positive map,
but not \mbox{2-positive}, remains open.
It should be noted that from \cite[Theorem 3.3]{marciniak2008extremal}
we know that every extremal 2-positive map is completely positive.

\begin{theorem}
\label{thm:ExtremalSchwarz}
Let $S: M_{n} \rightarrow M_{n}$ be a Schwarz map that is extremal in the set
of all positive maps.
Suppose additionally, that there exist
rank-one orthogonal projections $P, Q \in M_{n}$, for which $S(P) = Q$.
Then $S$ is completely positive. 
\end{theorem}

\begin{proof}
Let $(e_{i})_{i=1}^{n}$ be the standard orthonormal basis in $\mathbb{C}^{n}$.  
Because we want to prove that $S$ is completely positive,
we can safely assume that $P = Q = P_{n}$,
where $P_{n} = e_{n} e_{n}^{*}$ is the orthogonal projection on the one-dimensional space
spanned by the vector $e_{n} = (0,0,\ldots,0,1) \in \mathbb{C}^{n}$. 
From the fact that $S$ preserves the identity follows that
$S(P_{n}^{\perp}) = P_{n}^{\perp}$,
where $P_{n}^{\perp} = \mathbf{1}_{n} - P_{n}$. 
For a matrix $X \in M_{n}$,
written in the block form
$X = \left( \begin{smallmatrix} A & u \\ w^{t} & z \end{smallmatrix} \right)$,
where $A \in M_{n-1}$, $u, w \in \mathbb{C}^{n-1}$ are column vectors,
and $z \in \mathbb{C}$, we have
\begin{equation}
\label{eq:SBlockForm}
    S (X) = \begin{pmatrix}
        \hat{S}_{0}(A) & S_{1} u \\
        (\overline{S}_{1} w)^{t} & z
    \end{pmatrix},
\end{equation}
where $\hat{S}_{0}: M_{n-1} \rightarrow M_{n-1}$ is necessarily a Schwarz map
of $M_{n-1}$, and $S_{1} \in M_{n-1}$.
Indeed, 
since $S$ is a Schwarz map, 
by \cite[Proposition 2.1.5]{stormer2012positive},
$S(P_{n} X) = P_{n} S(X)$.
Now, it follows that
$S(P_{n} X P_{n}) = P_{n} S(X) P_{n}$,
and similarly $S(P_{n}^{\perp} X P_{n}^{\perp}) = P_{n}^{\perp} S(X) P_{n}^{\perp}$,
$S(P_{n}^{\perp} X P_{n}) = P_{n}^{\perp} S(X) P_{n}$,
$S(P_{n} X P_{n}^{\perp}) = P_{n} S(X) P_{n}^{\perp}$,
which shows that $S$ must have the form \eqref{eq:SBlockForm}.

If we take a one-dimensional orthogonal projection 
$P_{u} = u u^{*}$, where
$u = (u_{1}, u_{2}, \ldots, u_{n}) = (\vec{u}, u_{n})$,
$\vec{u} = (u_{1}, u_{2}, \ldots, u_{n-1}) \in \mathbb{C}^{n-1}$,
$u_{n} \in \mathbb{C}$,
we have then
\begin{equation}
 SP_{u} \:=\:  S \begin{pmatrix}
    \vec{u} \vec{u}^{\,*} & \overline{u}_{n} \vec{u} \\
    u_{n} \vec{u}^{\,*}   & |u_{n}|^{2}
 \end{pmatrix} \: = \: 
 \begin{pmatrix}
    \hat{S}_{0}(\vec{u} \vec{u}^{*}) &
         \overline{u}_{n} S_{1} \vec{u} \\
    u_{n} ( S_{1} \vec{u} )^{*} &
        |u_{n}|^{2}
 \end{pmatrix}.
\end{equation}
In the case when $u_{n} \neq 0$,
taking the Schur complement
(see \cite[Theorem 1.12, p.34]{zhang2006schur}),
we have that $SP_{u} \geq 0$, if and only if
\begin{equation}
\label{ieq:SchurForS}
  S_{1} \vec{u} \vec{u}^{*} S_{1}^{*} \leq \hat{S}_{0}(\vec{u}\vec{u}^{*}).
\end{equation}
But that means that $S$ is a sum of two positive maps that act on a one-dimensional
orthogonal projection by:
\begin{equation}
 SP_{u} \:=\:  
 \begin{pmatrix}
      S_{1} \vec{u} \vec{u}^{*} S_{1}^{*}  &
         \overline{u}_{n} S_{1} \vec{u} \\
    u_{n} ( S_{1} \vec{u} )^{*} &
        |u_{n}|^{2}
 \end{pmatrix} +   
 \begin{pmatrix}
    \hat{S}_{0}(\vec{u} \vec{u}^{*})  -  S_{1} \vec{u} \vec{u}^{*} S_{1}^{*} & 0 \\
    0 & 0
 \end{pmatrix}.
\end{equation}
Since $S$ is extremal and $S(P_{n}) = P_{n}$,
\begin{equation}
 SP_{u} \:=\:  
 \begin{pmatrix}
      S_{1} \vec{u} \vec{u}^{*} S_{1}^{*}  &
         \overline{u}_{n} S_{1} \vec{u} \\
    u_{n} ( S_{1} \vec{u} )^{*} &
        |u_{n}|^{2}
 \end{pmatrix} = 
    \begin{pmatrix}
    S_{1} & 0 \\ 0 & 1
    \end{pmatrix}
    \, P_{u} \,
    \begin{pmatrix}
    S_{1} & 0 \\ 0 & 1
    \end{pmatrix}^{*}.
\end{equation}
If we put $U = \left( \begin{smallmatrix} S_{1} & 0 \\ 0 & 1 \end{smallmatrix} \right) \in M_{n}$,
then by the fact that $S(P_{n}^{\perp}) = P_{n}^{\perp}$,
we have that $U U^{*} = \mathbf{1}_{n}$, i.e. $U$ is unitary,
and $S(X) = U X U^{*}$, $X \in M_{n}$, which completes the proof.
\end{proof}

Having presented geometrical arguments for the decomposition theorem
of positive maps of $M_{2}$ in Theorem \ref{thm:PositiveMaps},
we pass now to a proposition concerning bistochastic maps.
We choose a tentative symbol $\Delta_{\infty}$ for the set of all bistochastic map of $M_{2}$.
Let also $\Delta_{1} \subset \Delta_{\infty}$ denote the set of those bistocastic maps
that are both completely positive and completely copositive at the same time.
Using the same geometric picture of positive maps as operators
preserving the Lorentz cone $L_{4}$,
we can show that the set $\Delta_{1}$ is surprisingly large.
It is worth mentioning that some related results could be found in 
\cite{stormer2013decomposition}.

Any matrix $Y \in M_{n}(\mathbb{R})$ admits the singular value decomposition:
$Y = O_{1} D O_{2}$, where
$O_{1}, O_{2}$ are orthogonal matrices
and $D = \text{diag}(t_{1}, t_{2}, \ldots, t_{n})$ is a diagonal matrix
such that the \emph{singular values} of $Y$ are ordered:
$t_{1} \geq t_{2} \geq \ldots t_{n} \geq 0$.
The singular values of $Y$ are the eigenvalues of the matrix
$|Y| = (Y^{t} Y)^{1/2}$.
In the space of operators $M_{n}(\mathbb{R})$,
we distinguish two norms.
The first one is the usual operator norm: $|| \cdot ||$;
the second one is the trace norm given by
$||Y||_{1} = \sqrt{\text{Tr} \, Y^{t} \, Y }$.
We denote the unit ball of $M_{3}(\mathbb{R})$ in the operator norm by:
$K_{\infty} = \left \{ Y \in M_{3}(\mathbb{R}): \, ||Y|| \leq 1 \right \}$;
and the ball in the trace norm by:
$K_{1} = \left \{ Y \in M_{3}(\mathbb{R}): \, ||Y||_{1} \leq 1 \right \}$.

\begin{theorem}
\label{thm:Ball}
There is a convex isomorphism $F$ between the set of bistochastic maps of $M_{2}$
and the set $K_{\infty}$, $F(\Delta_{\infty}) = K_{\infty}$. 
Moreover, $F(\Delta_{1}) = K_{1}$.    
\end{theorem}

\begin{proof}
Let $S$ be a bistochastic map.
It is true that in that case
\begin{equation}
\label{eq:Sbistochastic}
\pi(S) = \begin{pmatrix}
    1  &  \vec{0}^{t} \\
    \vec{0} & Y
    \end{pmatrix},
\end{equation}
where $Y \in M_{3}(\mathbb{R})$
and $\vec{0} \in \mathbb{R}^{3}$ is a null column vector.
We define $F(S) = Y$. 
From the definition \eqref{RandomLabel:365828} of the map $\pi$,
it is evident that for $S_{1}, S_{2} \in \Delta_{\infty}$,
we have 
$F \left( \lambda S_{1} + (1-\lambda) S_{2} \right) = 
\lambda F(S_{1}) + (1-\lambda) F(S_{2})$,
$0 \leq \lambda \leq 1$,  
and $F(S_{1} S_{2}) = F(S_{1}) F(S_{2})$.
The map $S$ is positive, if and only if
$\pi(S)$ preserves the cone $L_{4}$,
which in this case is equivalent to the fact that $F(S) = Y \in K_{\infty}$.
Moreover, it is easy to see that both $F$ and $F^{-1}$ are continuous mapping
with respect to the topologies induced by the natural norms on   
$\Delta_{\infty}$ and $K_{\infty}$.
Hence, $F(\Delta_{\infty}) = K_{\infty}$ and in addition, $F$ is a homeomorphism.

Suppose now that $S \in \Delta_{1}$.
Since $S$ is completely positive,
from \cite[Theorem 1]{landau1993birkhoff},
we know that $S$ is a convex combination of unitary maps:
$X \mapsto U^{*} X U$, $U \in \text{U}(2)$.
From the proof of \mbox{Theorem \ref{thm:PositiveMaps}},
for the particular form of $\pi(S)$ as in \eqref{eq:Sbistochastic},
it means that $Y$ belongs to the convex hull of the group $\text{SO}(3)$.
From \cite{miranda1994group}, Corollary on p. 139,
we infer that if $t_{1} \geq t_{2} \geq t_{3} \geq 0$
are singular values of the matrix $Y$,
then $Y \in \text{conv} \, \text{SO}(3)$,
if and only if
\begin{equation}
\label{eq:ineqForSingVals}
t_{1} \leq 1
\quad \text{and} \quad
t_{1} + t_{2} - \text{sig} (\text{det} Y) \, t_{3} \leq 1.
\end{equation}
$S$ is also completely copositive,
i.e. $S' = S \circ t$ is a completely positive bistochastic map.
Let $F(S') = Y'$.
Then $Y' = Y \, \text{diag} (1,-1,1)$.
Both matrices, $Y$ and $Y'$, share the same singular values and
$\text{det} \, Y' = - \text{det} \, Y$.
Because $Y' \in  \text{conv} \, \text{SO}(3)$,
we have that
$t_{1} + t_{2} + \text{sig} (\text{det} Y) \, t_{3} \leq 1$.
Combining this with \eqref{eq:ineqForSingVals},
since $\text{det} \, Y = 0$, if and only if $t_{3} = 0$,
we obtain that the map $S$ is both completely positive
and completely copositive, if and only if
$||Y||_{1} = t_{1} + t_{2} + t_{3} \leq 1$,
which altogether means that $F(\Delta_{1}) = K_{1}$. 
\end{proof}

To summarise, we have presented a geometrical approach that proves itself useful
in the investigation of properties of positive maps on $M_{2}$.
In particular, we have been able to show that the decomposition theorem for
positive maps on $M_{2}$ is a direct consequence of the fact that every operator
preserving the Lorentz cone is a convex combination of the ones that
in addition preserve the boundary of the cone.
These operators translate as completely positive and completely copositive maps,
and therefore one is able to obtain an elegant and simple proof of the decomposition
theorem that leads e.g. to the celebrated PPT criterion for separable states
of two-qubit systems \cite{peres1996separability,horodecki1996separability}.
The proof of Proposition \ref{prop:MapsPreservingIdentity},
although it does not make use of this geometrical approach explicitly,
is backed by the insight into the structure of the operators that preserve the
Lorentz cone and have norm equal to one.
Hence the assumption \ref{lem:condProj} in Proposition \ref{prop:MapsPreservingIdentity}.
Finally, Theorem \ref{thm:Ball} shows that the set of 
bistochastic maps that are completely positive and completely copositive
is unexpectedly large and topologically equivalent to the unit ball
of $M_{3}(\mathbb{R})$ with respect to the trace norm.


\bibliographystyle{unsrt}
\bibliography{./biblio}

\begin{thebibliography}{10}

\bibitem{stormer2012positive}
E.~St{\o}rmer.
\newblock {\em {Positive Linear Maps of Operator Algebras}}.
\newblock Springer monographs in mathematics. Springer, 2012.

\bibitem{chruscinski2014entanglement}
D.~Chru{\'s}ci{\'n}ski and G.~Sarbicki.
\newblock Entanglement witnesses: construction, analysis and classification.
\newblock {\em Journal of Physics A: Mathematical and Theoretical},
  47(48):483001, 2014.

\bibitem{stormer1963positive}
E.~St{\o}rmer.
\newblock Positive linear maps of operator algebras.
\newblock {\em Acta Mathematica}, 110(1):233--278, 1963.

\bibitem{woronowicz1976positive}
S.L. Woronowicz.
\newblock Positive maps of low dimensional matrix algebras.
\newblock {\em Reports on Mathematical Physics}, 10(2):165--183, 1976.

\bibitem{stinespring1955positive}
W.~F. Stinespring.
\newblock {Positive functions on C*-algebras}.
\newblock {\em Proceedings of the American Mathematical Society},
  6(2):211--216, 1955.

\bibitem{choi1975completely}
M.-D. Choi.
\newblock Completely positive linear maps on complex matrices.
\newblock {\em Linear algebra and its applications}, 10(3):285--290, 1975.

\bibitem{peres1996separability}
A.~Peres.
\newblock Separability criterion for density matrices.
\newblock {\em Physical Review Letters}, 77(8):1413--1415, 1996.

\bibitem{horodecki1996separability}
M.~Horodecki, P.~Horodecki, and R.~Horodecki.
\newblock Separability of mixed states: necessary and sufficient conditions.
\newblock {\em Physics Letters A}, 223(1-2):1--8, 1996.

\bibitem{loewy1975positive}
R.~Loewy and H.~Schneider.
\newblock Positive operators on the n-dimensional ice cream cone.
\newblock {\em Journal of Mathematical Analysis and Applications},
  49(2):375--392, 1975.

\bibitem{leinaas2006geometrical}
Jon~Magne Leinaas, Jan Myrheim, and Eirik Ovrum.
\newblock Geometrical aspects of entanglement.
\newblock {\em Physical Review A}, 74(1):012313, 2006.

\bibitem{robertson1983schwarz}
A.~G. Robertson.
\newblock {Schwarz inequalities and the decomposition of positive maps on
  C*-algebras}.
\newblock In {\em Mathematical Proceedings of the Cambridge Philosophical
  Society}, volume~94, pages 291--296. Cambridge Univ Press, 1983.

\bibitem{choi1974schwarz}
M.-D. Choi.
\newblock {A Schwarz inequality for positive linear maps on C*-algebras}.
\newblock {\em Illinois Journal of Mathematics}, 18(4):565--574, 1974.

\bibitem{naber1992geometry}
G.L. Naber.
\newblock {\em The Geometry of Minkowski Spacetime: An Introduction to the
  Mathematics of the Special Theory of Relativity}.
\newblock Applied mathematics sciences. Dover, 1992.

\bibitem{goyal2011geometry}
K.~Goyal, S.\, N.~Simon, B.\, R.~Singh, and S.~Simon.
\newblock {Geometry of the generalized Bloch sphere for qutrit}.
\newblock {\em arXiv preprint arXiv:1111.4427}, 2011.

\bibitem{besenyei2011completely}
Ad{\'a}m Besenyei and D{\'e}nes Petz.
\newblock Completely positive mappings and mean matrices.
\newblock {\em Linear Algebra and its Applications}, 435(5):984--997, 2011.

\bibitem{marciniak2008extremal}
M.~Marciniak.
\newblock {On extremal positive maps acting between type I factors}.
\newblock {\em arXiv preprint arXiv:0812.2311}, 2008.

\bibitem{zhang2006schur}
F.~Zhang.
\newblock {\em The Schur complement and its applications}.
\newblock Springer, 2006.

\bibitem{stormer2013decomposition}
Erling St{\o}rmer.
\newblock A decomposition theorem for positive maps, and the projection onto a
  spin factor.
\newblock {\em arXiv preprint arXiv:1308.3332}, 2013.

\bibitem{landau1993birkhoff}
L.~J. Landau and R.~F. Streater.
\newblock {On Birkhoff's theorem for doubly stochastic completely positive maps
  of matrix algebras}.
\newblock {\em Linear algebra and its applications}, 193:107--127, 1993.

\bibitem{miranda1994group}
F.~Miranda, H.\ and R.~C. Thompson.
\newblock Group majorization, the convex hulls of sets of matrices, and the
  diagonal element-singular value inequalities.
\newblock {\em Linear algebra and its applications}, 199:131--141, 1994.

\end{thebibliography}

\end{document}